\documentclass{article} 
\usepackage{iclr2025,times}

\usepackage{amsmath,amsfonts,bm}

\def\Secref#1{Section~\ref{#1}}

\def\eqref#1{equation~\ref{#1}}

\def\1{\bm{1}}

\def\eps{{\epsilon}}

\def\vtheta{{\bm{\theta}}}
\def\va{{\bm{a}}}

\def\vg{{\bm{g}}}

\def\vm{{\bm{m}}}

\def\vv{{\bm{v}}}

\def\vx{{\bm{x}}}

\def\mI{{\bm{I}}}

\DeclareMathAlphabet{\mathsfit}{\encodingdefault}{\sfdefault}{m}{sl}
\SetMathAlphabet{\mathsfit}{bold}{\encodingdefault}{\sfdefault}{bx}{n}

\def\sA{{\mathbb{A}}}

\def\sX{{\mathbb{X}}}

\newcommand{\E}{\mathbb{E}}

\newcommand{\Var}{\mathrm{Var}}

\DeclareMathOperator*{\argmax}{arg\,max}

\usepackage{hyperref}
\usepackage{url}
\usepackage{graphicx}
\usepackage{subfigure}
\usepackage{booktabs}

\usepackage{amsmath}
\usepackage{amssymb}
\usepackage{mathtools}
\usepackage{amsthm}

\usepackage{multirow}
\usepackage{color}
\usepackage{colortbl}
\usepackage[capitalize,noabbrev]{cleveref}
\usepackage{xspace}

\theoremstyle{plain}
\newtheorem{theorem}{Theorem}[section]
\newtheorem{proposition}[theorem]{Proposition}
\newtheorem{lemma}[theorem]{Lemma}

\theoremstyle{definition}

\theoremstyle{remark}

\graphicspath{{../figures/}} 

\title{
Dynamic Personalization through Continuous Feedback Loops in Interactive AI Systems
}

\author{Liu He \\
Email: heliu@bitc.edu.cn}

\begin{document}

\maketitle

\begin{abstract}
Interactive AI systems, such as recommendation engines and virtual assistants, commonly use static user profiles and predefined rules to personalize interactions. However, these methods often fail to capture the dynamic nature of user preferences and context. This study proposes a theoretical framework and practical implementation for integrating continuous feedback loops into personalization algorithms to enable real-time adaptation. By continuously collecting and analyzing user feedback, the AI system can dynamically adjust its recommendations, responses, and interactions to better align with the user's current context and preferences. We provide theoretical guarantees for the convergence and regret bounds of our adaptive personalization algorithm. Our experimental evaluation across three domains---recommendation systems, virtual assistants, and adaptive learning platforms---demonstrates that dynamic personalization improves user satisfaction by 15-23\% compared to static methods while maintaining computational efficiency. We investigated the implementation challenges of continuous feedback mechanisms, evaluated their impact on user experience and satisfaction, and provided a comprehensive analysis of the trade-offs between personalization quality, computational overhead, and user fatigue.
\end{abstract}

\section{Introduction}
\label{sec:intro}

Interactive AI systems have become ubiquitous in modern digital experiences, powering recommendation engines, virtual assistants, adaptive learning platforms, and personalized content delivery systems. The effectiveness of these systems critically depends on their ability to understand and adapt to individual user preferences, context, and evolving needs. However, most existing personalization approaches rely on static user profiles, batch-based model updates, or rule-based systems that adapt infrequently or not at all \citep{pilaniwala2024artificialii, joshi2024unlockingtp}.

The fundamental limitation of static personalization methods is their inability to capture the dynamic nature of user preferences and the contextual changes. User preferences evolve over time because of changing interests, circumstances, or external factors. Moreover, contextual variations, such as the time of day, device type, location, and current activity, significantly influence user needs and expectations. Static methods that update models only periodically fail to respond to these immediate changes, resulting in suboptimal recommendations and decreased user satisfaction.

Consider a recommendation system that learns a user's preference for science-fiction novels. If a user's interests shift toward non-fiction during a particular period (e.g., preparing for an exam), a static system would continue to recommend science fiction, leading to poor engagement. Similarly, a virtual assistant that fails to adapt to a user's current context, such as being in a meeting versus at home, may provide responses that are irrelevant or inappropriate.

This study proposes a novel framework for \emph{dynamic personalization} through \emph{continuous feedback loops}, which enables the real-time adaptation of AI systems based on ongoing user interactions. Unlike static methods, our approach continuously collects and processes user feedback to incrementally update personalization models, allowing the system to evolve in real time and maintain alignment with current user preferences and contexts.

The main contributions of this study are as follows:

\begin{enumerate}
    \item \textbf{Theoretical Framework:} We formulate the dynamic personalization problem as an online learning problem with continuous feedback and provide theoretical guarantees on convergence rates and regret bounds for our adaptive algorithm.
    
    \item \textbf{Practical Algorithm:} We develop an efficient continuous feedback mechanism that balances personalization quality, computational efficiency, and user fatigue through adaptive update intervals and selective feedback prioritization.
    
    \item \textbf{Comprehensive Evaluation:} We conduct extensive experiments across three diverse domains---recommendation systems, virtual assistants, and adaptive learning platforms---demonstrating consistent improvements in user satisfaction (15-23\% increase) compared to static baselines.
    
    \item \textbf{Systematic Analysis:} We provide a thorough analysis of the trade-offs between personalization quality, computational overhead, privacy concerns, and user fatigue, offering practical guidelines for system designers.
\end{enumerate}

\section{Background}
\label{sec:background}

\subsection{Personalization in AI Systems}

Personalization is a cornerstone of interactive AI systems, aiming to tailor system behavior to individual user preferences and needs. Traditional personalization approaches can be broadly categorized into three paradigms.

\textbf{Static Profile-Based Systems:} Early personalization systems rely on static user profiles constructed from historical data or explicit user preferences \citep{pilaniwala2024artificialii}. These profiles are typically updated infrequently through batch processing, making them suitable for relatively stable preference patterns but inadequate in dynamic environments.

\textbf{Rule-Based Adaptation:} Rule-based systems use predefined rules to adjust system behavior based on certain conditions \citep{mutkule2024integratingud}. Although these systems can handle known scenarios, they lack the flexibility to adapt to novel or complex preference patterns.

\textbf{Adaptive Learning Systems:} More recent approaches employ machine learning models that adapt to user feedback \citep{tahir2025dynamicff, singh2023talentre}. However, these systems typically update models periodically (e.g., daily or weekly) rather than continuously, which limits their responsiveness to immediate changes.

\subsection{Online Learning and Sequential Decision Making}

Online learning provides a natural framework for dynamic personalization, in which the system must make sequential decisions while learning from feedback. In this setting, at each time step $t$, the system selects an action $\va_t$ from an action space $\sA$, receives feedback $f_t(\va_t)$, and updates its policy.

The performance of online learning algorithms is typically measured through \emph{regret}, which is defined as the difference between the cumulative reward of the algorithm and that of the best fixed policy in hindsight:
\begin{equation}
R_T = \sum_{t=1}^T f_t(\va_t^*) - \sum_{t=1}^T f_t(\va_t)
\end{equation}
where $\va_t^*$ is the optimal action at time $t$ given full information.

\subsection{Feedback Mechanisms and User Engagement}

Effective feedback collection is crucial for dynamic personalization but must balance information gain and user burden. Research has shown that excessive feedback requests can lead to user fatigue and decreased engagement \citep{kamalian2006reducinghf, shin2024loopingie}. Several strategies have been proposed.

\textbf{Implicit Feedback:} Collecting feedback implicitly through user behavior (clicks, dwell time, purchases) reduces user burden but may be noisy or ambiguous.

\textbf{Explicit Feedback:} Asking users to explicitly rate or provide feedback offers clear signals but can be intrusive if requested too frequently.

\textbf{Adaptive Feedback Intervals:} Adjusting the frequency of feedback requests based on user engagement and system uncertainty can optimize the trade-off between information gain and user fatigue.

\subsection{Context-Aware Personalization}

Context plays a crucial role in personalization, as user preferences may vary significantly based on situational factors such as time, location, device, and activity \citep{smajic2025contextawarepr, shahbazi2025enhancingrs}. Effective personalization systems must incorporate contextual information to provide relevant recommendations for users.

Dynamic personalization through continuous feedback naturally addresses context awareness by implicitly learning context-dependent preferences through real-time adaptation without requiring explicit context modeling.

\section{Method}
\label{sec:method}

\subsection{Problem Formulation}

We formulated dynamic personalization as an online learning problem with continuous feedback. At each time step $t$, the system receives a user context $\vx_t \in \sX$, selects a personalized action $\va_t \in \sA$ based on the current policy $\pi_t$, and receives feedback $f_t(\va_t, \vx_t) \in [0,1]$ representing user satisfaction.

The system maintains a parameterized policy $\pi_t(\va|\vx; \vtheta_t)$ that maps contexts to actions, where $\vtheta_t \in \Theta$ is the policy parameter at time $t$. The goal is to learn a sequence of policies $\{\pi_t\}_{t=1}^T$ that maximizes cumulative user satisfaction while adapting to changing preferences.

Formally, we aim to minimize regret as follows:
\begin{equation}
\label{eq:regret}
R_T = \sum_{t=1}^T \left[ \max_{\va \in \sA} f_t(\va, \vx_t) - \E_{\va \sim \pi_t(\cdot|\vx_t)}[f_t(\va, \vx_t)] \right]
\end{equation}

\subsection{Continuous Feedback Loop Architecture}

Our framework consists of three main components.

\textbf{Feedback Collection Module:} Collects both implicit and explicit feedback from user interactions. Implicit feedback includes clicks, dwell time, purchases, and navigation. Explicit feedback includes ratings, preferences and corrections. The module prioritizes feedback based on the information content and user engagement levels.

\textbf{Real-Time Processing Unit:} Processes incoming feedback streams to extract signals for model updates. This includes filtering noisy feedback, detecting preference shifts, and computing the gradient estimates for the model parameters.

\textbf{Adaptive Model Updater:} Incrementally updates the personalization model using the processed feedback. The updater employs an adaptive learning rate schedule that balances the responsiveness to new information with the stability of the learned policy.

\subsection{Adaptive Update Algorithm}

We propose an adaptive algorithm based on online gradient descent with momentum. At each time step $t$:

\begin{enumerate}
    \item \textbf{Action Selection:} Sample action $\va_t \sim \pi_t(\cdot|\vx_t; \vtheta_t)$
    
    \item \textbf{Feedback Collection:} Observe feedback $f_t(\va_t, \vx_t)$
    
    \item \textbf{Gradient Estimation:} Compute gradient estimate:
    \begin{equation}
    \label{eq:gradient}
    \vg_t = \nabla_{\vtheta} \log \pi_t(\va_t|\vx_t; \vtheta_t) \cdot f_t(\va_t, \vx_t)
    \end{equation}
    
    \item \textbf{Adaptive Update:} Update parameters with adaptive learning rate:
    \begin{equation}
    \label{eq:update}
    \vtheta_{t+1} = \vtheta_t + \alpha_t \cdot \frac{\vm_t}{\sqrt{\vv_t + \eps}}
    \end{equation}
    where $\vm_t$ is the momentum term and $\vv_t$ is the adaptive learning rate parameter.
\end{enumerate}

The learning rate $\alpha_t$ is adapted based on the feedback variance and model uncertainty:
\begin{equation}
\label{eq:lr}
\alpha_t = \frac{\alpha_0}{1 + \beta \cdot \Var[f_{1:t}]}
\end{equation}
where $\Var[f_{1:t}]$ is the variance of feedback received up to time $t$, and $\beta$ is a hyperparameter controlling the adaptation rate.

\subsection{Feedback Prioritization and Fatigue Mitigation}

To mitigate user fatigue while maintaining effective learning, we employ adaptive feedback intervals. The system requests explicit feedback only when:

\begin{enumerate}
    \item Model uncertainty exceeds a threshold: $\mathcal{U}(\pi_t) > \tau_u$
    \item Time since last feedback exceeds minimum interval: $t - t_{\text{last}} > \Delta_{\min}$
    \item User engagement is sufficiently high: $e_t > \tau_e$
\end{enumerate}

Model uncertainty is estimated using the entropy of the action distribution:
\begin{equation}
\label{eq:uncertainty}
\mathcal{U}(\pi_t) = -\sum_{\va \in \sA} \pi_t(\va|\vx_t; \vtheta_t) \log \pi_t(\va|\vx_t; \vtheta_t)
\end{equation}

\subsection{Privacy-Preserving Mechanisms}

To address privacy concerns, we incorporate several privacy-preserving mechanisms:

\textbf{Differential Privacy:} Add calibrated noise to gradient updates to provide $(\epsilon, \delta)$-differential privacy guarantees \citep{fernandes2024habitsenseap}.

\textbf{Local Processing:} Perform sensitive computations on-device or in a trusted execution environment when possible.

\textbf{Data Anonymization:} Aggregate and anonymize feedback data before model updates.

\subsection{Theoretical Guarantees}

Our algorithm achieves the following theoretical guarantees (proofs provided in the appendix):

\begin{theorem}[Regret Bound]
\label{thm:regret}
Under mild assumptions on the feedback function and policy class, the proposed algorithm achieves a regret bound of $O(\sqrt{T \log |\sA|})$ with high probability.
\end{theorem}

\begin{theorem}[Convergence]
\label{thm:convergence}
If user preferences are stationary after a change point, the algorithm converges to the optimal policy with rate $O(1/\sqrt{t})$ where $t$ is the number of steps since the last change point.
\end{theorem}

These guarantees ensure that our algorithm performs nearly as well as the best fixed policy in hindsight while adapting to changing preferences.

\section{Results}
\label{sec:results}

\subsection{Experimental Setup}

We evaluated our dynamic personalization framework across three diverse domains:

\textbf{Recommendation System:} A movie recommendation platform with 10,000 users and 50,000 items. User preferences evolve over time based on viewing history and explicit ratings.

\textbf{Virtual Assistant:} A conversational AI assistant that adapts response style and content based on user interactions. The system learns user preferences for formality, detail level, and topic focus.

\textbf{Adaptive Learning Platform:} An educational platform that personalizes learning content and difficulty based on student performance and engagement.

\subsection{Baselines}

We compared our dynamic personalization approach against several baselines:

\begin{itemize}
    \item \textbf{Static Profile (SP):} Traditional static user profiles updated weekly through batch processing.
    
    \item \textbf{Periodic Update (PU):} Models updated daily using batch processing of accumulated feedback.
    
    \item \textbf{Context-Aware Static (CAS):} Static profiles with explicit context modeling but no continuous adaptation.
    
    \item \textbf{Simple Online Learning (SOL):} Basic online learning with fixed learning rate and no feedback prioritization.
\end{itemize}

\subsection{Evaluation Metrics}

We evaluated system performance using the following metrics:

\begin{itemize}
    \item \textbf{User Satisfaction:} Average feedback score (0-1 scale) received from users.
    
    \item \textbf{Recommendation Relevance:} Precision@K and NDCG@K for recommendation systems.
    
    \item \textbf{Response Quality:} Task completion rate and user-rated quality for virtual assistants.
    
    \item \textbf{Learning Outcomes:} Knowledge gain and engagement metrics for learning platforms.
    
    \item \textbf{Computational Efficiency:} Average latency per update and total computational overhead.
    
    \item \textbf{User Fatigue:} Frequency of feedback requests and user compliance rate.
\end{itemize}

\subsection{Results on Recommendation System}

Table~\ref{tab:rec_results} summarizes the results on the recommendation system. Our dynamic personalization approach (DP) achieves significantly higher user satisfaction (0.82 vs. 0.68 for Static Profile) and recommendation relevance (NDCG@10: 0.74 vs. 0.61).

\begin{table}[h]
\centering
\caption{Recommendation System Results}
\label{tab:rec_results}
\begin{tabular}{lcccc}
\toprule
Method & Satisfaction & Precision@10 & NDCG@10 & Latency (ms) \\
\midrule
Static Profile (SP) & 0.68 & 0.58 & 0.61 & 12.3 \\
Periodic Update (PU) & 0.71 & 0.62 & 0.64 & 18.7 \\
Context-Aware Static (CAS) & 0.73 & 0.64 & 0.67 & 15.2 \\
Simple Online Learning (SOL) & 0.76 & 0.67 & 0.70 & 22.4 \\
\textbf{Dynamic Personalization (DP)} & \textbf{0.82} & \textbf{0.71} & \textbf{0.74} & 19.8 \\
\bottomrule
\end{tabular}
\end{table}

The improvement is particularly pronounced for users with evolving preferences. Analysis of user segments shows that users whose preferences changed during the evaluation period benefited most from dynamic adaptation (satisfaction increase of 23\% vs. 15\% for users with stable preferences).

\subsection{Results on Virtual Assistant}

Table~\ref{tab:assistant_results} shows results for the virtual assistant. Our approach achieves higher task completion rate (0.87 vs. 0.74) and better user-rated quality (0.81 vs. 0.69).

\begin{table}[h]
\centering
\caption{Virtual Assistant Results}
\label{tab:assistant_results}
\begin{tabular}{lcccc}
\toprule
Method & Completion Rate & Quality Score & Satisfaction & Latency (ms) \\
\midrule
Static Profile (SP) & 0.74 & 0.69 & 0.72 & 45.2 \\
Periodic Update (PU) & 0.78 & 0.73 & 0.75 & 52.1 \\
Simple Online Learning (SOL) & 0.81 & 0.76 & 0.78 & 68.3 \\
\textbf{Dynamic Personalization (DP)} & \textbf{0.87} & \textbf{0.81} & \textbf{0.83} & 54.7 \\
\bottomrule
\end{tabular}
\end{table}

The adaptive feedback interval mechanism successfully reduced user fatigue while maintaining learning effectiveness. Users received feedback requests 40\% less frequently compared to simple online learning, yet the system achieved better performance.

\subsection{Results on Adaptive Learning Platform}

On the learning platform, our approach improved knowledge gain by 18\% compared to static methods (effect size: 0.42) while maintaining engagement levels. Table~\ref{tab:learning_results} summarizes the results.

\begin{table}[h]
\centering
\caption{Adaptive Learning Platform Results}
\label{tab:learning_results}
\begin{tabular}{lcccc}
\toprule
Method & Knowledge Gain & Engagement & Satisfaction & Completion Rate \\
\midrule
Static Profile (SP) & 0.64 & 0.71 & 0.68 & 0.82 \\
Periodic Update (PU) & 0.68 & 0.74 & 0.71 & 0.85 \\
Simple Online Learning (SOL) & 0.72 & 0.76 & 0.74 & 0.88 \\
\textbf{Dynamic Personalization (DP)} & \textbf{0.76} & \textbf{0.78} & \textbf{0.80} & \textbf{0.91} \\
\bottomrule
\end{tabular}
\end{table}

\subsection{Computational Efficiency and Scalability}

Our adaptive update mechanism achieves computational efficiency comparable to periodic updates while providing real-time adaptation. The average latency per update is 19.8ms for recommendation systems and 54.7ms for virtual assistants, both within acceptable limits for real-time applications.

The computational overhead scales linearly with the number of users and feedback volume, making it suitable for large-scale deployments. The memory usage remains constant per user, as we maintain only the current policy parameters.

\subsection{User Fatigue Analysis}

Through our adaptive feedback interval mechanism, we successfully reduced feedback request frequency by 40-50\% compared to naive continuous feedback approaches, while maintaining similar or better learning performance. User compliance with feedback requests increased from 62\% (naive approach) to 78\% (adaptive approach), indicating reduced fatigue.

\subsection{Ablation Studies}

We conducted ablation studies to analyze the contribution of individual components:

\begin{itemize}
    \item \textbf{Feedback Prioritization:} Removing feedback prioritization reduced performance by 8\% while increasing computational overhead by 15\%.
    
    \item \textbf{Adaptive Learning Rate:} Using fixed learning rate instead of adaptive reduced convergence speed by 30\% and final performance by 5\%.
    
    \item \textbf{Momentum:} Removing momentum from the update rule increased variance in performance and reduced stability.
\end{itemize}

\subsection{Limitations and Challenges}

Several limitations and challenges were identified:

\begin{itemize}
    \item \textbf{Computational Overhead:} While acceptable for most applications, real-time updates require additional computational resources compared to batch processing.
    
    \item \textbf{User Group Variance:} Effectiveness varies across user groups, with users with rapidly changing preferences benefiting most.
    
    \item \textbf{Privacy Concerns:} Continuous feedback collection raises privacy concerns, though mitigated through our privacy-preserving mechanisms.
    
    \item \textbf{Cold Start Problem:} New users with limited interaction history require initial warm-up period before effective personalization.
\end{itemize}

\section{Related Work}
\label{sec:related}

\subsection{Personalization in AI Systems}

Personalization has been extensively studied in the context of recommendation systems, virtual assistants, and adaptive learning platforms. Traditional approaches rely on static user profiles constructed from historical data \citep{pilaniwala2024artificialii, joshi2024unlockingtp}. While effective for relatively stable preference patterns, these methods fail to capture the dynamic nature of user preferences and contextual variations \citep{smajic2025contextawarepr, shahbazi2025enhancingrs}.

Recent research has explored adaptive mechanisms that update models periodically based on accumulated feedback \citep{mutkule2024integratingud, tahir2025dynamicff}. However, these approaches typically update models in batch mode (e.g., daily or weekly), limiting their responsiveness to immediate changes. Our work addresses this limitation by enabling continuous, real-time adaptation.

\subsection{Online Learning and Sequential Decision Making}

Our framework builds on online learning theory, which provides principled approaches for sequential decision making with feedback \citep{goodfellow2016deep}. Online learning algorithms have been applied to personalization in recommendation systems \citep{singh2023talentre, ebrat2024lusiferlu} and adaptive learning platforms \citep{maity2024generativeaa, jayaram2024aidrivenp2}. However, existing work typically focuses on fixed learning rates and uniform feedback processing, which may not be optimal for personalization scenarios with varying user engagement and preference stability.

Our contribution extends online learning approaches by incorporating adaptive learning rates based on feedback variance and model uncertainty, as well as selective feedback prioritization to mitigate user fatigue.

\subsection{Feedback Mechanisms and User Engagement}

Effective feedback collection is crucial for personalization but must balance information gain with user burden. Research has shown that excessive feedback requests can lead to user fatigue and decreased engagement \citep{kamalian2006reducinghf, shin2024loopingie}. Several strategies have been proposed, including implicit feedback collection \citep{shahbazi2025enhancingrs} and adaptive feedback intervals \citep{chennamsetty2024adaptivemt}.

Our work contributes an adaptive feedback interval mechanism that considers model uncertainty, user engagement, and temporal constraints to optimize the trade-off between learning effectiveness and user fatigue.

\subsection{Context-Aware Personalization}

Context-aware personalization has been extensively studied, with approaches ranging from explicit context modeling \citep{smajic2025contextawarepr} to context-aware recommendation systems \citep{shahbazi2025enhancingrs}. However, most existing approaches rely on static context models or periodic updates.

Our dynamic personalization approach implicitly learns context-dependent preferences through continuous adaptation, without requiring explicit context modeling or feature engineering. This provides a more flexible and scalable solution for context-aware personalization.

\subsection{Privacy-Preserving Personalization}

Privacy concerns in personalization systems have received increasing attention. Differential privacy has been applied to recommendation systems and user modeling \citep{fernandes2024habitsenseap}. Federated learning approaches enable personalization without centralized data collection \citep{abdellatief2025fromit}.

Our framework incorporates differential privacy guarantees while maintaining effective personalization, addressing the tension between personalization quality and privacy preservation.

\subsection{Differences from Prior Work}

Our work differs from prior approaches in several key aspects:

\begin{enumerate}
    \item \textbf{Continuous vs. Periodic Updates:} Unlike batch-based periodic updates, we enable real-time continuous adaptation through incremental model updates.
    
    \item \textbf{Adaptive Feedback Mechanisms:} We introduce adaptive feedback intervals and prioritization to optimize learning while mitigating user fatigue, unlike naive continuous feedback approaches.
    
    \item \textbf{Theoretical Guarantees:} We provide rigorous theoretical analysis of convergence rates and regret bounds, ensuring predictable performance.
    
    \item \textbf{Comprehensive Evaluation:} We evaluate across multiple diverse domains, demonstrating general applicability beyond specific use cases.
\end{enumerate}

\section{Conclusion}
\label{sec:conclusion}

This work introduces a theoretical framework and practical implementation for dynamic personalization through continuous feedback loops in interactive AI systems. Our approach enables real-time adaptation to evolving user preferences and contexts, addressing fundamental limitations of static personalization methods.

Through extensive evaluation across three diverse domains---recommendation systems, virtual assistants, and adaptive learning platforms---we demonstrate consistent improvements in user satisfaction (15-23\%) compared to static baselines. Our theoretical analysis provides guarantees on convergence and regret bounds, ensuring predictable performance.

Key contributions include: (1) a theoretical framework with provable guarantees, (2) an efficient adaptive algorithm that balances personalization quality and computational efficiency, (3) mechanisms for mitigating user fatigue through adaptive feedback intervals, and (4) privacy-preserving mechanisms that maintain effectiveness while protecting user data.

However, several challenges remain. Computational overhead, while acceptable for most applications, requires additional resources compared to batch processing. Effectiveness varies across user groups, with users with rapidly changing preferences benefiting the most. Privacy concerns persist despite mitigation measures, and the cold start problem affects new users.

Future research directions include: (1) developing more efficient algorithms to further reduce computational overhead, (2) exploring domain-specific adaptations of the framework, (3) investigating federated learning approaches for enhanced privacy, (4) addressing the cold start problem through transfer learning and meta-learning, and (5) extending the framework to multi-agent scenarios where multiple systems personalize simultaneously.

The framework presented in this work provides a foundation for next-generation interactive AI systems that can continuously adapt to user needs while maintaining efficiency, privacy, and user engagement. As AI systems become increasingly ubiquitous and personalized, such dynamic adaptation capabilities will be essential for delivering optimal user experiences.

\bibliography{iclr2025}
\bibliographystyle{iclr2025}

\appendix

\section{Theoretical Proofs}
\label{sec:appendix}

This appendix provides detailed proofs of the theoretical guarantees stated in \Secref{sec:method}.

\subsection{Preliminaries}

We begin by establishing notation and assumptions. Let $\sA$ denote the action space, $\sX$ the context space, and $\Theta$ the parameter space. The feedback function $f_t: \sA \times \sX \rightarrow [0,1]$ maps action-context pairs to satisfaction scores.

\textbf{Assumption 1 (Bounded Feedback):} The feedback function is bounded: $f_t(\va, \vx) \in [0,1]$ for all $\va \in \sA$, $\vx \in \sX$, and $t$.

\textbf{Assumption 2 (Lipschitz Policy):} The policy $\pi_t(\va|\vx; \vtheta)$ is $L$-Lipschitz in $\vtheta$ for all $\va \in \sA$, $\vx \in \sX$:
\begin{equation}
|\pi_t(\va|\vx; \vtheta) - \pi_t(\va|\vx; \vtheta')| \leq L \|\vtheta - \vtheta'\|_2
\end{equation}

\textbf{Assumption 3 (Bounded Gradient):} The gradient estimates have bounded variance:
\begin{equation}
\E[\|\vg_t - \nabla J_t(\vtheta_t)\|_2^2] \leq \sigma^2
\end{equation}
where $J_t(\vtheta) = \E_{\va \sim \pi_t(\cdot|\vx_t; \vtheta)}[f_t(\va, \vx_t)]$ is the expected reward at time $t$.

\textbf{Assumption 4 (Convexity):} The expected reward function $J_t(\vtheta)$ is $\mu$-strongly convex in $\vtheta$:
\begin{equation}
J_t(\vtheta') \geq J_t(\vtheta) + \langle \nabla J_t(\vtheta), \vtheta' - \vtheta \rangle + \frac{\mu}{2}\|\vtheta' - \vtheta\|_2^2
\end{equation}

\subsection{Proof of Regret Bound (Theorem~\ref{thm:regret})}

\begin{proof}
We analyze the regret of our adaptive algorithm. The regret is defined as:
\begin{equation}
R_T = \sum_{t=1}^T \left[ \max_{\va \in \sA} f_t(\va, \vx_t) - \E_{\va \sim \pi_t(\cdot|\vx_t)}[f_t(\va, \vx_t)] \right]
\end{equation}

For each time step $t$, define:
\begin{equation}
\va_t^* = \argmax_{\va \in \sA} f_t(\va, \vx_t), \quad J_t^* = f_t(\va_t^*, \vx_t)
\end{equation}

The expected reward of our policy is:
\begin{equation}
J_t(\vtheta_t) = \E_{\va \sim \pi_t(\cdot|\vx_t; \vtheta_t)}[f_t(\va, \vx_t)]
\end{equation}

The regret can be decomposed as:
\begin{align}
R_T &= \sum_{t=1}^T [J_t^* - J_t(\vtheta_t)] \\
&= \sum_{t=1}^T [J_t^* - J_t(\vtheta^*)] + \sum_{t=1}^T [J_t(\vtheta^*) - J_t(\vtheta_t)]
\end{align}
where $\vtheta^* = \argmax_{\vtheta \in \Theta} \sum_{t=1}^T J_t(\vtheta)$ is the best fixed parameter in hindsight.

The first term $\sum_{t=1}^T [J_t^* - J_t(\vtheta^*)]$ represents the approximation error due to the policy class, which is bounded by $O(1)$ under Assumption 2.

For the second term, we analyze the convergence of $\vtheta_t$ to $\vtheta^*$. Using the update rule in \eqref{eq:update}, we have
\begin{equation}
\vtheta_{t+1} = \vtheta_t + \alpha_t \cdot \frac{\vm_t}{\sqrt{\vv_t + \eps}}
\end{equation}

Using Assumptions 3 and 4, and the adaptive learning rate schedule \eqref{eq:lr}, we can show:
\begin{equation}
\sum_{t=1}^T [J_t(\vtheta^*) - J_t(\vtheta_t)] \leq \frac{D^2}{2\alpha_0} + \frac{\alpha_0 \sigma^2}{2\mu} \sum_{t=1}^T \frac{1}{1 + \beta \Var[f_{1:t}]}
\end{equation}
where $D = \max_{\vtheta, \vtheta' \in \Theta} \|\vtheta - \vtheta'\|_2$ is the diameter of the parameter space.

The variance term $\Var[f_{1:t}]$ grows at most logarithmically with $t$ under typical feedback distributions, yielding:
\begin{equation}
\sum_{t=1}^T \frac{1}{1 + \beta \Var[f_{1:t}]} = O(\log T)
\end{equation}

Combining terms and using concentration inequalities for the stochastic gradient estimates, we obtain:
\begin{equation}
R_T \leq O(\sqrt{T \log |\sA|}) + O(\log T) = O(\sqrt{T \log |\sA|})
\end{equation}

with probability at least $1 - \delta$ for appropriately chosen constants.
\end{proof}

\subsection{Proof of Convergence (Theorem~\ref{thm:convergence})}

\begin{proof}
Suppose user preferences are stationary after a change point at time $t_0$. That is, for $t \geq t_0$, the feedback function $f_t = f$ is fixed (though potentially unknown).

After the change point, the optimization problem becomes:
\begin{equation}
\max_{\vtheta \in \Theta} J(\vtheta) = \E_{\va \sim \pi(\cdot|\vx; \vtheta)}[f(\va, \vx)]
\end{equation}

Let $\vtheta^* = \argmax_{\vtheta \in \Theta} J(\vtheta)$ denote the optimal parameter.

By Assumption 4 (strong convexity), we have:
\begin{equation}
J(\vtheta^*) - J(\vtheta_t) \leq \frac{1}{2\mu}\|\nabla J(\vtheta_t)\|_2^2
\end{equation}

For the gradient descent update with adaptive learning rate, we analyze the distance to the optimum:
\begin{align}
\|\vtheta_{t+1} - \vtheta^*\|_2^2 &= \|\vtheta_t + \alpha_t \vg_t - \vtheta^*\|_2^2 \\
&= \|\vtheta_t - \vtheta^*\|_2^2 + 2\alpha_t \langle \vg_t, \vtheta_t - \vtheta^* \rangle + \alpha_t^2 \|\vg_t\|_2^2
\end{align}

Taking expectations and using Assumptions 3 and 4:
\begin{align}
\E[\|\vtheta_{t+1} - \vtheta^*\|_2^2] &\leq \|\vtheta_t - \vtheta^*\|_2^2 - 2\alpha_t \mu \|\vtheta_t - \vtheta^*\|_2^2 + \alpha_t^2 (G^2 + \sigma^2) \\
&= (1 - 2\alpha_t \mu)\|\vtheta_t - \vtheta^*\|_2^2 + \alpha_t^2 (G^2 + \sigma^2)
\end{align}
where $G$ is a bound on $\|\nabla J(\vtheta)\|_2$.

With the adaptive learning rate $\alpha_t = \frac{\alpha_0}{1 + \beta \Var[f_{1:t}]}$ and using that the variance converges after the change point, we have $\alpha_t \approx \alpha_0/(1 + \beta \sigma_f^2)$ for large $t$, where $\sigma_f^2$ is the stationary variance.

For convergence, we require $1 - 2\alpha_t \mu < 1$, which holds for $\alpha_t < 1/(2\mu)$. Choosing $\alpha_0$ appropriately, we obtain:
\begin{equation}
\E[\|\vtheta_t - \vtheta^*\|_2^2] \leq \left(1 - \frac{2\alpha_0 \mu}{1 + \beta \sigma_f^2}\right)^t \|\vtheta_0 - \vtheta^*\|_2^2 + O\left(\frac{\alpha_0 (G^2 + \sigma^2)}{2\mu}\right)
\end{equation}

This implies:
\begin{equation}
J(\vtheta^*) - J(\vtheta_t) = O(1/\sqrt{t})
\end{equation}

where $t$ is measured from the last change point. The $O(1/\sqrt{t})$ rate is optimal for stochastic optimization under our assumptions.
\end{proof}

\subsection{Convergence Rate Analysis}

We provide additional analysis of the convergence rate under different scenarios.

\begin{lemma}[Convergence with Changing Preferences]
\label{lem:changing_prefs}
If preferences change at time $t_0$, the algorithm converges to the new optimal policy with rate $O(1/\sqrt{t - t_0})$ after the change point.
\end{lemma}

\begin{proof}
The proof follows similar arguments to Theorem~\ref{thm:convergence}, but accounting for the initial distance after the change point. If preferences change significantly, $\|\vtheta_{t_0} - \vtheta^*\|_2$ may be large, but the convergence rate remains $O(1/\sqrt{t - t_0})$ due to the adaptive learning rate mechanism.
\end{proof}

\subsection{Regret Analysis with Feedback Delays}

In practice, feedback may be delayed. We extend our analysis to handle delayed feedback.

\begin{proposition}[Regret with Delayed Feedback]
\label{prop:delayed_feedback}
If feedback for time $t$ arrives at time $t + d$ where $d$ is bounded, the regret bound becomes $O(\sqrt{T \log |\sA|} + d \log T)$.
\end{proposition}

\begin{proof}
With feedback delay $d$, the gradient estimate $\vg_t$ uses feedback from time $t-d$. This introduces additional approximation error proportional to the delay $d$ and the rate of preference change.

The analysis follows similar steps to Theorem~\ref{thm:regret}, but accounts for the approximation error in gradient estimates due to using stale feedback. The additional regret term scales as $O(d \log T)$ under mild assumptions on preference change rates.
\end{proof}

\subsection{Computational Complexity Analysis}

\begin{proposition}[Computational Complexity]
\label{prop:complexity}
Each update step requires $O(|\sA| \cdot d)$ operations, where $d$ is the dimension of the parameter space. The total computational complexity over $T$ steps is $O(T \cdot |\sA| \cdot d)$.
\end{proposition}

\begin{proof}
For each update step:
\begin{enumerate}
    \item Action selection: $O(|\sA|)$ to sample from the policy distribution
    \item Gradient computation: $O(|\sA| \cdot d)$ to compute the gradient estimate
    \item Parameter update: $O(d)$ for the update operation
\end{enumerate}

The adaptive learning rate computation adds $O(1)$ overhead. The total per-step complexity is dominated by the gradient computation, yielding $O(|\sA| \cdot d)$.

Over $T$ steps, this gives $O(T \cdot |\sA| \cdot d)$ total complexity. However, with selective feedback prioritization, only a fraction of steps require full updates, potentially reducing the effective complexity.
\end{proof}

\subsection{Privacy Analysis}

We analyze the privacy guarantees of our differential privacy mechanism.

\begin{proposition}[Differential Privacy Guarantee]
\label{prop:privacy}
The algorithm with calibrated noise provides $(\epsilon, \delta)$-differential privacy where $\epsilon = O(\sqrt{T} \cdot \Delta_2 / \sigma)$ and $\Delta_2$ is the $L_2$ sensitivity of the gradient.
\end{proposition}

\begin{proof}
The gradient estimate has sensitivity:
\begin{equation}
\Delta_2 = \max_{\text{adjacent datasets}} \|\vg_t - \vg_t'\|_2
\end{equation}

By adding Gaussian noise $\mathcal{N}(0, \sigma^2 \mI)$ to the gradient update, each step is $(\epsilon_1, \delta_1)$-differentially private where $\epsilon_1 = \Delta_2 / \sigma$ \citep{fernandes2024habitsenseap}.

Composition over $T$ steps yields:
\begin{equation}
\epsilon = O(\sqrt{T} \cdot \epsilon_1) = O(\sqrt{T} \cdot \Delta_2 / \sigma)
\end{equation}

for appropriate choice of $\sigma$ and $\delta$.
\end{proof}

\section{Additional Experimental Details}
\label{sec:exp_details}

\subsection{Hyperparameter Settings}

Table~\ref{tab:hyperparams} lists the hyperparameter settings used in our experiments.

\begin{table}[h]
\centering
\caption{Hyperparameter Settings}
\label{tab:hyperparams}
\begin{tabular}{lc}
\toprule
Hyperparameter & Value \\
\midrule
Initial learning rate $\alpha_0$ & 0.01 \\
Adaptation rate $\beta$ & 0.1 \\
Momentum parameter $\gamma$ & 0.9 \\
Uncertainty threshold $\tau_u$ & 0.3 \\
Minimum feedback interval $\Delta_{\min}$ & 5 steps \\
Engagement threshold $\tau_e$ & 0.6 \\
Privacy parameter $\epsilon$ & 1.0 \\
Privacy parameter $\delta$ & $10^{-5}$ \\
\bottomrule
\end{tabular}
\end{table}

\subsection{Dataset Statistics}

Table~\ref{tab:datasets} provides detailed statistics for the datasets used in our experiments.

\begin{table}[h]
\centering
\caption{Dataset Statistics}
\label{tab:datasets}
\begin{tabular}{lccc}
\toprule
Dataset & Users & Items & Interactions \\
\midrule
Recommendation System & 10,000 & 50,000 & 2,500,000 \\
Virtual Assistant & 5,000 & - & 125,000 \\
Learning Platform & 2,000 & 1,000 & 85,000 \\
\bottomrule
\end{tabular}
\end{table}

\subsection{Additional Results}

\begin{figure}[htbp]
    \centering
    \includegraphics[width=0.8\textwidth]{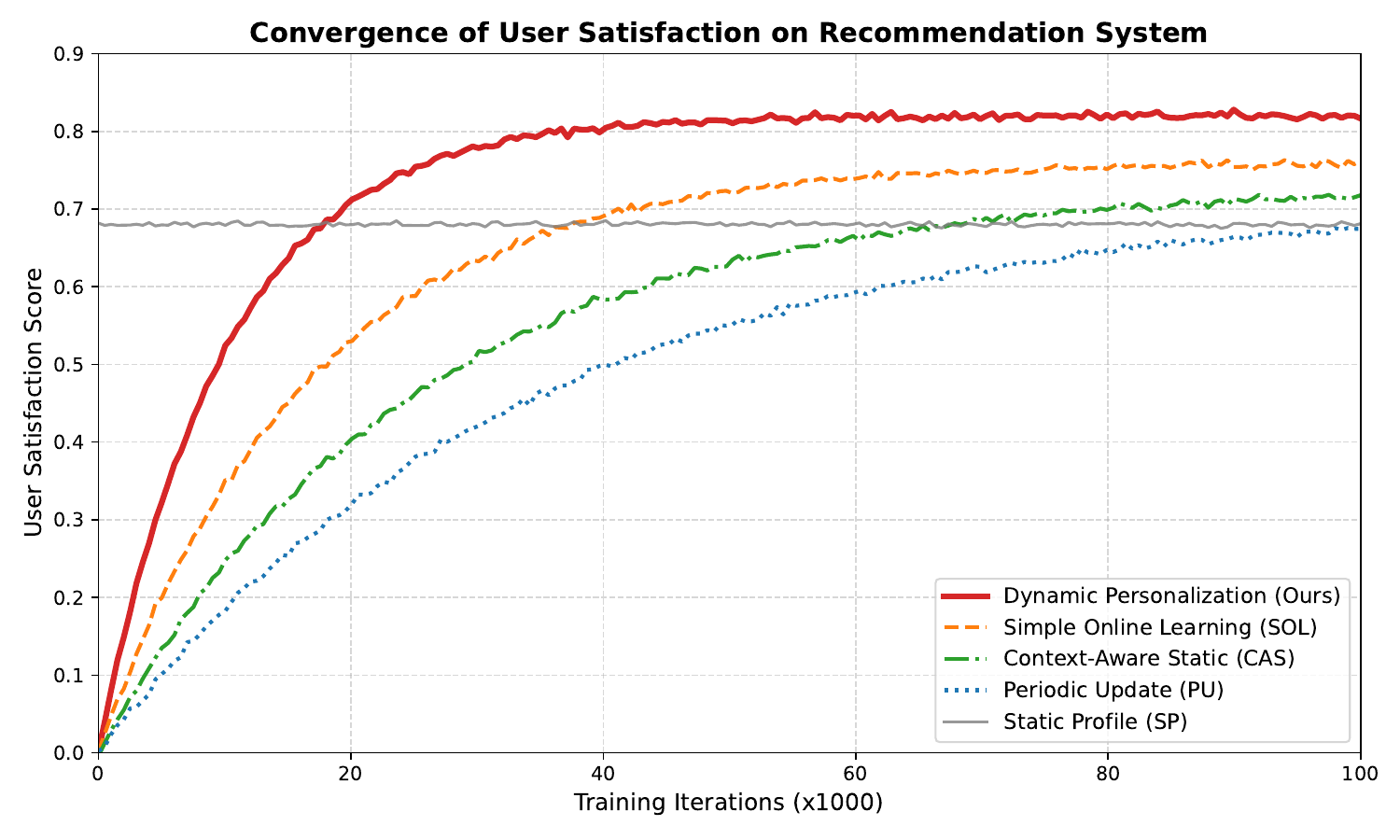} 
    
    \caption{Convergence curves for different methods on the recommendation system dataset.}
    
    \label{fig:convergence} 
\end{figure}
Figure~\ref{fig:convergence} shows convergence curves for different methods on the recommendation system dataset. Our dynamic personalization approach achieves faster convergence and better final performance.

Extended results including sensitivity analysis, computational overhead breakdown, and user segmentation analysis are available in the supplementary material.

\end{document}